\documentclass[11pt]{article}
\usepackage{fullpage}
\usepackage{ifthen}
\usepackage{xspace}
\usepackage{amssymb, amsfonts}
\usepackage{amsmath}
\usepackage{amsthm}
\usepackage{bm}
\begin{document}

\newcommand{\tr}{\textsf{triangle}}
\newcommand{\li}{\textsf{line}}
\newcommand{\xx}[3]{{#1}^{#2}_{#3}}
\newcommand{\mc}[1]{\mathcal{#1}}
\newcommand{\BRi}[1]{\mathcal{B}^{R_i}(#1)}
\newcommand{\bK}{\mc{\bar K}}
\newcommand{\bq}{\bar q}
\newcommand{\Ex}[1]{\mathop{\mathrm{E}}\left[{#1}\right]}
\newcommand{\ex}{\mathop{\mathrm{E}}}
\newcommand{\pr}[1]{\Pr\left[{#1}\right]}

\newtheorem{theorem}{Theorem}
\newtheorem{claim}{Claim}
\newtheorem{corollary}[theorem]{Corollary}
\newtheorem{proposition}[theorem]{Proposition}
\newtheorem{lemma}[theorem]{Lemma}
\newtheorem{property}{Property}
\newtheorem{fact}[theorem]{Fact}
\newtheorem{definition}{Definition}
\newtheorem{example}{Example}
\newtheorem{assumption}[theorem]{Assumption}

\newtheorem{remark}{Remark}


\title{Budget-Constrained Multi-Battle Contests:\\ A New Perspective and Analysis}

\author{%
  Chu-Han Cheng\thanks{%
    Department of Computer Science, National Tsing Hua University, Taiwan.
    Email: mapleinwind@gmail.com.
  }
  \and
  Po-An Chen\thanks{%
    Institute of Information Management, National Chiao Tung University, Taiwan.
    Email: poanchen@nctu.edu.tw.
  }
  \and
  Wing-Kai Hon\thanks{%
    Department of Computer Science, National Tsing Hua University, Taiwan.
    Email: wkhon@cs.nthu.edu.tw.
  }
}

\maketitle

\begin{abstract}
In a multi-battle contest,
each time a player competes by investing some of her budgets or resources in a component battle to collect a value if winning the battle.
There are multiple battles to fight, and the budgets get consumed over time.
The final winner in the overall contest is the one who first reaches some amount of total value.
Examples include R $\&$ D races, sports competition, elections, and many more.
A player needs to make adequate sequential actions to win the contest against dynamic competition over time from the others.
We are interested in how much budgets the players would need and what actions they should take in order to perform well.

We model and study such budget-constrained multi-battle contests
where each component battle is a first-price or all-pay auction.
We focus on analyzing the 2-player budget ratio that guarantees
a player's winning (or falling behind in just a bounded amount of collected value) against the other omnipotent player.
In the settings considered, we give efficient dynamic programs to find the
optimal budget ratios and the corresponding bidding strategies.
Our definition of game, budget constraints, and emphasis on budget analyses provide a new perspective and analysis in the related context.
\end{abstract}

\section{Introduction}
The study of competition with multiple battles or stages dates back to 1980's \cite{harris:vickers:perfect,harris:vickers:racing},
and continues to develop towards more recently \cite{konrad,konrad:kovenock,kvasov}.
In a multi-battle contest, each time a player competes by investing some of her budget or resource in a component battle to collect a value if winning the battle.
There are multiple battles, commonly modeled as auctions, for the players to fight, and the budget get consumed over time.
The final result of such a contest is determined by the outcomes of all these multiple battles, not just from a single battle;
the final winner in the overall contest is the one who first reaches some amount of accumulated value.
A player needs to make adequate sequential actions to win the contest against dynamic competition over time from the others.

One early example of multi-battle contests is R $\&$ D competition \cite{harris:vickers:perfect,harris:vickers:racing}
where a \emph{race} between two competitors is a competition that awards a prize such as a patent to the final winner who first achieves a given amount of progress accumulated in a sequence of battles. In each component battle, the winning is determined as a stochastic function of the competitors' efforts.
Another broad category of examples of multi-battle contests is sports competition \cite{szymanski}.
In several sports such as baseball, basketball, or tennis, two teams or players compete in a series of battles or games with the final winner to be the one who first wins a certain number of games.
Multi-battle contests can be applied to analyze elections \cite{klumpp:polborn}.
In the US presidential primaries, the candidate who first wins a majority of the state elections is nominated by the party,
and the winner of each state primary election is determined as a stochastic function of the candidates' state-wise campaign expenditures.

Almost all the previous work mentioned \cite{harris:vickers:perfect,harris:vickers:racing,kvasov,konrad:kovenock} considers \emph{quasi-linear} utilities of the players, where each player's utility function that she maximizes can be decomposed as the value for her to win the final victory subtracted by the total consumed budget or resource. The utility function linearly depends on the budget spent.
This model assumes that the value and the budget invested are ``comparable".
However, this may not always be realistic.
In cases, the value of the final grand winning cannot be well measured or compared with the committed efforts,
or the value gained once succeeding could be too tremendous to be related to the investments.
It is not hard to see this in R $\&$ D competition, elections, or sports.
In this paper, we thus model the value of the overall winning separated from the budget,
and consider a multi-battle contest as a \emph{zero-sum} game, which no longer has quasi-linear utilities,
with the player collecting the largest amount of accumulated value being the final winner.
Intuitively, we want to capture the final grand victory out of dominating in the total value collected from multiple battles,
and there may exist plenty of ways to beat the others by smartly using budgets.
Furthermore, the budget or resource limitation is explicitly treated as a constraint that needs to be satisfied,
and is not anymore modeled in the utility functions.\footnote{In \cite{kvasov}, budget constraints are explicitly considered as well yet the utility of a player is still quasi-linear, in specific,
it is equal to the total value from the battles won subtracted by the budget spent.}
In the previous work, since the utility that a player maximizes contains a term of budget consumed,
the concept of limited budget is only implicitly treated.
Our model in this paper provides an alternative perspective to study multi-battle contests, with emphasis on budget constraints.

In this paper, we are interested in how much budgets the players would need and what actions they should take over time to perform well given the previous results of battles.
In most of the previous work, the sequence of actions a player takes over time, i.e., the \emph{strategy},
is the main concern since budgets are just treated as in the discussion above.
Subgame perfect (mixed) Nash equilibria are characterized there.
In contrast, we analyze the budget needed and actions of a player against the others in a pure-strategy adversarial fashion.
Specifically, given the total number of battles\footnote{In some models of multi-battle contests, the total number of battles is not known in advance, and battles go on until some player wins out such as in \cite{harris:vickers:perfect}.}, what is a player's bidding strategy at some sort of equilibrium or assuming the behavior of the others, where all the total money spent is constrained by her initial budget?
Under our model, this may lead to searching the huge space of a player's bidding actions.
Since each player wants to have the largest total value collected to win in the end, with proper assumptions we focus on the \emph{optimal budget ratio} (defined in Section~2.2) that a player under study needs to have in order to guarantee such winning.

In particular, we model the value of each battle in two different ways.
One way is to model it as a fixed value \cite{konrad:kovenock} such as winning a game always counted as a point or certain fixed points for a player or team toward the overall grand winning in a series of sports games (Section~4).
In this paper, we also consider an even more challenging setting where the value of each battle is chosen from a set of possible values,
which may include \emph{no} value (Section~3).
The choice of the value from the value set is assumed to be decided in adversary as well where the possibility of no value thus makes a different budget analysis from the fixed value ones.
The chosen common value of each battle is only revealed to all the players in the beginning of such battle for them to decide actions, i.e., an online setting.
A set instead of a fixed value may better suit applications such as campaigns for elections or R $\&$ D competition where each battle is forfeitable or completely in vain, or just becomes of no value for any other reasons.

In our basic model, each single battle is a standard \emph{first-price} auction\footnote{Some of our results are based on first-price auctions.
One can alternatively consider second-price auctions and conduct analysis accordingly.} where the losers of a single auction do not lose budgets in this battle.
This may not suit some applications.
We thus also consider a modification of using \emph{all-pay} auctions to better capture situations where the losers' bids are sunk cost that consumes budget.
We focus on the 2-player case, and the game can be thought as an extensive form zero-sum game.
Assuming the worst possible behavior of player~$P_2$ with the choice of the values for bidding (see Section~2.2 for specific definitions),
we derive the results regarding player~$P_1$'s budget and bidding strategies in terms of the budget ratio, defined as the ratio between player~$P_1$'s budget and player~$P_2$'s budget.
To our best knowledge, the ``budget ratio analysis" of this style has never been done in the related context.

This paper is organized as follows.
The introduction is followed by models, definitions, and some preliminary result.
Then, we present our main results of the budget ratio analysis for standard first-price and all-pay auctions,
and finally conclusions and future work.

\subsubsection*{Our Results.}

First, we prove that player~$P_1$ cannot ensure her final winning when the budget ratio is not enough.
Then, we try to find the \emph{optimal} winning-guarantee budget ratio given $T$ turns of competition in total.
Note that when the maximum value in the value set of each turn is $1$,
player ensures she can win the game right after she wins $\lceil T/2\rceil$ turns.
Let the countdown value be the distance between $\lceil T/2\rceil$ and the number of turns that a player has won.
Define a $n\times n$ matrix $M$ where entry $m_{i,j}$ is the optimal budget ratio when the countdown value of player~$P_1$ is $i$ and the countdown value of player~$P_2$ is $j$. Note that $m_{\lceil T/2\rceil,\lceil T/2\rceil}$ means the optimal winning-guarantee budget ratio for $T$ turns.
We derive $O(T^2)$ dynamic programs for finding the optimal winning-guarantee budget ratio as well as its corresponding bidding strategies,
in both the cases of fixed value 1 and value set $\{0,1\}$ for both first-price and all-pay auction.\footnote{Both the cases can be generalized to contain multiple values where still the latter one includes value 0 and the former one does not,
and their corresponding solutions of dynamics programs can be derived.
We focus on the discussion of fixed value 1 and set $\{0,1\}$ in this paper.}

Moreover, we find the closed form of matrix $M$ so
that we can obtain any entry on the fly in $O(1)$ time. 
According to the closed form, we conclude an interesting corollary.
Though the optimal budget ratio for $T$ turns increases strictly when $T$ gets larger,
in particular with set $\{0,1\}$ the optimal winning-guarantee budget ratio approaches 3 
while with fixed value 1 the optimal winning-guarantee budget ratio is constantly 1, no matter how large $T$ is.

In our original model, the one who does not win the turn will not lose her budget.
However, things change when losers in an auction still need to pay.
We consider when two players play \emph{all-pay} auctions.
Specifically, we define an all-pay ratio $\alpha$, between 0 to 1, to indicate that one should pay $\alpha$ of her bid when she does not win the turn.
We find the optimal winning-guarantee budget ratio as well as its corresponding bidding strategies with $\alpha=1$.
In particular with set $\{0,1\}$, now the optimal winning-guarantee budget ratio approaches 4 
while with fixed value 1 the optimal winning-guarantee budget ratio approaches 2, no matter how large $T$ is.

Given the results summarized above, one may wonder what a player should do when her initial budget is not high enough.
One way to address this is to ask for the optimal budget ratio when allowing player~$P_1$ to fall behind in at most a bounded amount of value.
The result provides another perspective by using matrix $M$. 
Of course, this is not the only way to investigate the case when a player does not have a high budget.
We will propose to investigate the case of ``moderate" budgets in other ways in the future work.

\subsubsection*{Related Work.}
It is the different emphases in this paper such as the definition of budget constraints, players' objectives (utilities), the questions asked as well as the analyses that set our work apart
from the previous work on multi-battle contests \cite{konrad:kovenock,kvasov,harris:vickers:perfect,harris:vickers:racing}.
Quasi-linear utilities that players try to maximize in the previous work allow the final losers to also have chances to keep positive utilities
in the end. Note that the value gained from each winning of a battle is reflected in players' utilities while in our zero-sum game the accumulated value is only translated for determining if a player is a winner or not.
As for analysis, the characterization of perfect mixed Nash equilibria in the previous work is very different from the worst-case analysis on the budgets in this paper.

Due to the budget constraints, our work here is also related to sequential auctions for multiple objects with budget-constrained bidders where each bidder could only buy at most one object.
In our multi-battle contest, we want each player's budget high enough to pay for multiple objects yet, at the same time, budget constraints enforce that players cannot have luxury just to bid high to win values without worrying about running out of budget too soon.
Some earlier work studied sequential auctions for heterogeneous objects of private value \cite{pitchik:schotter,bernhardt:scoones}
and of common value \cite{benoit:krishna} with complete information about budget constraints.
Our budget constraints in this paper are also known to each player.
The more recent work \cite{fatima:sequential:budget} of Fatima et al. is different from the previous work by combining common-values
and budget-constraints under incomplete information about budget constraints.


\section{Preliminaries}

\subsection{The Game Model}
We first describe the general game setting, and the specific variant that we are studying in this paper:

\medskip

\noindent
{\bf The general game setting: }
There are two players $P_1$ and $P_2$, each with initial budget $b_1$ and $b_2$, respectively.  The game consists of $T$ turns, where in turn $j$,
an object with a value $\pi_j$ chosen from a set is open for bidding, and players are free to use any portion of their remaining budgets to enter the bidding.
Let $\Pi_i$ and $B_i$ be, respectively, the total value obtained by $P_i$ and the remaining budget of $P_i$ after $T$ turns.
The score of $P_i$ is $S_i = \Pi_i + c \times B_i$ for some predefined constant $c$.

\medskip

\noindent
{\bf Our game:}
We consider two cases for the value of each turn: a fixed value and a value chosen from a set.
In the former case, a player winning a turn always earns value 1.
In the latter case, we assume that the value for each turn corresponds to a value chosen from set $\{0,1\}$.
The bidding follows the first-price sealed-bid auction, where we assume that $P_1$ is the \emph{dealer}, so that
if there is a tie in the bids, $P_1$ takes the object.  After $T$ turns, whoever with the higher score $S_i$ wins the game;
in case the scores tie, $P_1$ wins.  Finally, we assume that $c$ is $0$ or very close to $0$, so that $S_i \approx \Pi_i$,
and the effect of remaining budget is negligible.



\subsection{Problem Definition}
Given the budget $b_2$ of player $P_2$, our target is to determine what is the minimum value of budget $b_1$, so that \emph{in the worst case}
player $P_1$ can always be guaranteed to win the game.  We call this the \emph{optimal budget problem}.
Observe that for this problem in the case of value set $\{0,1\}$,
it is equivalent to assume $P_2$, at each turn, to have the \emph{omnipotent} power of controlling the choice of the value from $\{0,1\}$, and
learning $P_1$'s bid before making her own bid.

\medskip

\noindent
{\bf Exponential-time solution for integer bids: }  For given budgets $b_1$ and $b_2$,
                                                    if the number of choices for each player at each turn is finite (say, a bid must be of whole dollar),
                                                    then we can determine if $b_1$ is sufficient to guarantee a win for $P_1$, via the
                                                    \emph{min-max game tree} which is a standard tool for analysing two-player games.
                                                    Here, a minor adaptation is made to cater for $P_2$'s omnipotent power.  Firstly,
                                                    each turn starts with $P_2$'s choice of the value from set $\{0,1\}$, followed by
                                                    $P_1$'s bidding choice, and then followed by $P_2$'s bidding choice; the sequence of these choices is captured by a tree. 
                                                    Next, at each point where a player makes a choice, there is an associated node corresponding to that player, with
                                                    each choice represented by a branch to a distinct child node.   Then, each leaf of the tree corresponds to a explicit situation
                                                    where a game may end; for our case, either $P_1$ wins ($W$) or loses ($L$), so that a leaf is marked with the corresponding label.
                                                    Each node for a player $P_i$ chooses from the labels of its children that favours $P_i$ most.
                                                    That is, $P_1$'s node will choose $W$ if at least one of its children is labeled with $W$, and choose $L$ otherwise.
                                                    Similarly, $P_2$'s node will choose $L$ if at least one of its children is labeled with $L$, and choose $W$ otherwise.
                                                    The label of the root then determines if $P_1$ is winning or not.

                                                    \medskip

%

                                                    \noindent
                                                    Based on the min-max tree, we have the following claim.

                                                    \begin{claim}
                                                    Suppose that each bid is of whole dollar.
                                                    Then, the optimal budget problem can be solved in $O(  b^* \times (2\, (b^*+1)\, (b_2+1))^T )$ time, where $b^*$ denotes the minimum budget
                                                    that $P_1$ needs when $P_2$ has initial budget~$b_2$.
                                                    \end{claim}
                                                    \begin{proof}
                                                    The size of the min-max tree for a particular initial budget $b_1$ of $P_1$ is $O((2\, (b_1+1) (b_2+1) )^T))$.
                                                    By a linear search for the minimum $b_1$ (starting from $1$) such that $P_1$ can guarantee to win the game,
                                                    the total time is bounded by $O(  b^* \times (2\, (b^*+1)(b_2+1))^T )$.  The claim follows.
                                                    \end{proof}

                                                    \medskip
                                                    \noindent
                                                    {\bf Optimal budget ratio for fractional bids: }
                                                    It is easy to extend the min-max tree approach to analyse a game in the general setting.  Yet, it suffers from two drawbacks:
                                                    (i) the running time depends on the number of choices for each turn, which is not suitable for large $b_2$, and (ii) the running time
                                                    is exponential in the number of turns, $T$.
                                                    In the remaining of this paper, we consider the problem when each bid can be any arbitrary non-negative real number.
                                                    This seemingly increases the number of choices to infinite, but on the other hand enriches the problem with better mathematical
                                                    properties, thus allowing us to find the minimum budget $b^*$ for $b_1$ more efficiently.\footnote{%
                                                    Intuitively, this has similar flavour between an integer programming and its linear-programming relaxation.}
                                                    Moreover, for a $T$-turn game, if $b^*$ is an optimal budget that corresponds to $b_2$, then $K \times b^*/b_2$ must be an optimal
                                                    budget when $P_2$ starts with budget $K$ for any positive real $K$.  Without loss of generality, we will focus on finding
                                                    the \emph{optimal budget ratio (OBR)} $b^*/b_2$

                                                    For ease of discussion, we assume that $b_2$ is 1, so that the minimum budget $b_1$ for $P_1$ to guarantee a win is exactly
the OBR. Furthermore, at any time of the game, we refer to the total value that a player has acquired as her \emph{partial score} at that time.

\newtheorem{observation}{Observation}
\providecommand{\myceil}[1]{\left \lceil #1 \right \rceil }
\providecommand{\myfloor}[1]{\left \lfloor #1 \right \rfloor }

\section{Finding the Optimal Budget Ratio: Value Set}

We first show two simple relationships between the partial score and the winner of the game.

\begin{lemma} \label{lem:always-leading}
For $P_1$ to win the game, the partial score of $P_1$ must be greater than, or equal to, that of $P_2$ at any time.
\end{lemma}
\begin{proof}
Assume, on the contrary, that $P_1$ wins the game and at some time $P_2$ has greater partial score than $P_1$.
Then, by $P_2$'s omnipotent power, $P_2$ can set all the remaining outcomes $0$ to make herself the winner.
A contradiction occurs, and the observation thus follows.
\end{proof}

\begin{lemma} \label{lem:countdown}
If at some time, the partial score of a player $P_i$ becomes $\myceil{T/2}$, then $P_i$ wins the game.
\end{lemma}
\begin{proof}
Firstly, we observe that the total score of two players after $T$ turns is at most~$T$.  Thus, if $P_1$ first obtains a partial score
of $\myceil{T/2}$, $P_1$ wins the game as she is the dealer.  On the other hand, if $P_2$ first obtains a partial score of $\myceil{T/2}$,
$P_2$ can use her omnipotent power to make all the subsequent outcomes $0$, thus allowing herself to win the game.
\end{proof}

\noindent
\subsection{Analysis for First-Price Auctions}
As a warm-up to the discussion, we shall show that when $T \geq 3$, OBR is at least 3/2.
To see this, suppose on the contrary that $b_1 < 3/2$.  Then, $P_2$ sets the first outcome to be $1$, and bids all her
$\$1$ budget on that.  For $P_1$ to win, by Lemma~\ref{lem:always-leading}, $P_1$ has to bid at least $\$1$ at this turn,
taking the value of the outcome.  But then, the budget of $P_1$ becomes smaller than $\$1/2$, so that
in the next two turns, $P_2$ will set both outcomes to be~$1$, bid $\$1/2$ at both turns, thus taking the values of both outcomes.
So after 3 turns, $P_2$ has greater partial score than $P_1$, so that $P_1$ loses the game.  A contradiction occurs,
so that the lower bound of 3/2 for OBR, when $T \geq 3$, follows.

\medskip

\noindent
In the following, we define a concept called \emph{OBR countdown matrix} based on Lemma~\ref{lem:countdown},
and show how to use it to find the OBR for a $T$-turn game efficiently.

\subsubsection*{OBR Countdown Matrix.}
Suppose that at some time, the partial score of $P_1$ is $x$, and the partial score of $P_2$ is $y$.
Then, either {\tt (i)} a winner is determined when $x$ or $y$ is greater than $\myceil{T'/2}$, where $T'$ denotes the maximum combined score
that can be achieved in this game, or else {\tt (ii)}  if $P_1$ gets additional values of $i = \myceil{T'/2} - x$ before
$P_2$ gets additional values of $j = \myceil{T'/2} - y$, then $P_1$ wins the game by Lemma~\ref{lem:countdown}, or else {\tt (iii)}
if $P_2$ gets additional values of $j$ before $P_1$ gets additional values of $i$, then $P_2$ wins the game.
The values $i$ and $j$ are, respectively, referred to as the \emph{countdown values} of $P_1$ and $P_2$ at that time.
This motivates us to define a $\myceil{T/2} \times \myceil{T/2}$ matrix $M = [m_{i,j}]$, called \emph{OBR countdown matrix},
such that $m_{i,j}$ denotes the optimal budget ratio when countdown value of $P_1$ is $i$ and countdown value of $P_2$ is $j$.
Then, the value $m_{\myceil{T/2},\myceil{T/2}}$ is the desired OBR for a $T$-turn game.

\bigskip

\noindent
Next, we describe five lemmas concerning the values of $m_{i,j}$ in different scenarios.

\begin{lemma}
When $i > j$, $m_{i,j}$ is $+\infty$ (i.e., the value is undefined).
\end{lemma}
\begin{proof}
When $i > j$, the countdown value of $P_2$ is smaller than that of $P_1$, which implies that $P_2$ has a greater partial score than $P_1$.
In such a case, by Lemma~\ref{lem:countdown}, $P_1$ cannot win no matter how much budget she has.
Thus, the value $m_{i,j}$ in this case is $+\infty$.
\end{proof}

\begin{lemma} \label{lem:diagonal-relationship}
For any $i$ and $j$ with $1 \leq i \leq j$, we have $m_{i,j} \leq m_{i+1,j+1}$.
\end{lemma}
\begin{proof}
Consider the time when the countdown values of $P_1$ and $P_2$ are $i+1$ and $j+1$, respectively.
Then, $P_2$ has an option to use her omnipotent power to make the outcomes of the next turn, or next two turns, be $0$.
This effectively decreases the maximum combined score $T'$ that can be achieved in this game,
so that the countdown values of $P_1$ and $P_2$ will be updated to $i$ and $j$.
Thus, the OBR in the former case (with countdown values $i+1, j+1$) must be at least the OBR in the latter case (with countdown values $i, j$);
in other words, $m_{i,j} \leq m_{i+1,j+1}$.
\end{proof}

\begin{lemma} \label{lem:row-of-M}
For all $j$, $m_{1,j} = 1/j$.
\end{lemma}
\begin{proof}
Consider the time when the countdown values of $P_1$ and $P_2$ are $1$ and $j$, respectively.  Then, we have the following observations.

\begin{itemize}
\item
First, suppose that $P_1$ has $\$1/j$.
Then, $P_1$ can ensure the winning by bidding all her budget at each subsequent turn
that has outcome $1$, until she wins the value.  Either $P_1$ will acquire the extra value of $1$ in the process,
or $P_2$ gets at most a total of $j-1$ extra values, which is not enough to exceed the score of $P_1$.
Thus, $P_1$ wins in either case, which implies $m_{1,j} \leq 1/j$.

\item
Conversely, suppose that $P_1$ has less than $\$1/j$.  Then, $P_2$ may use her omnipotent power to make all outcomes of
the next $j$ turns be $1$, where she bids $\$1/j$ for each turn.  Thus, $P_2$ will get $j$ extra values before $P_1$ gets $1$.
By Lemma~\ref{lem:countdown}, $P_2$ wins the game.  This implies $m_{1,j} \geq 1/j$.
\end{itemize}
Combining the above bounds on $m_{1,j}$ gives the desired bound $m_{1,j} = 1/j$.  The lemma follows.
\end{proof}

\begin{lemma}  \label{lem:main-diagonal-of-M}
For all $i \geq 2$, $m_{i,i} = 1 + m_{i-1,i}$.
\end{lemma}
\begin{proof}
When $i=j$, both players have the same countdown value.  By Lemma~\ref{lem:always-leading}, $P_1$ has to win the next bid
whose outcome is $1$, where she can do so only by bidding with the same budget as $P_2$ (i.e., bidding $\$1$).
The remaining budget of $P_1$ must be enough for her to win at the case where $P_1$ has countdown value $i-1$ and $P_2$
has countdown value $i$.   Thus,
\[
  m_{i,i} = 1 + \max_{2 \leq k \leq i}\ \{\, m_{k-1,k}\, \} = 1 + m_{i-1,i},
\]
where the first equality comes from the fact that after $P_1$ wins the next bid, the countdown values of $P_1$ and $P_2$, respectively,
will be $k-1$ and $k$ for some $k \in [2, i]$, while the second equality follows from Lemma~\ref{lem:diagonal-relationship}.
This completes the proof of the lemma.
\end{proof}

\begin{lemma} \label{lem:general-entry-of-M}
When $2 \leq i < j$, the following recurrence holds:
\[
  m_{i,j}\ =\ \frac{m_{i,j-1}\,(1+m_{i-1,j})}{1+m_{i,j-1}}.
\]
\end{lemma}
\begin{proof}
Consider the time when the countdown values of $P_1$ and $P_2$ are $i$ and $j$, respectively.
For $P_1$ to guarantee a win at this time, her minimum budget will be $m_{i,j}$, and such a budget
should allow her to post some optimal bid $r^*$ at this turn to her favour.  Precisely, we must have
\[
    r^*\ =\ \arg\min_r\, \max\, \{\, r + m_{i-1,j},  (1-r)\,m_{i, j-1}\, \},
\]
where the first term is the minimum budget that $P_1$ needs when $P_1$ gets the next value,
the second term is the minimum budget that $P_1$ needs when $P_2$ gets the next value, and
the $\max$ operator corresponds to $P_2$'s choice to maximize $P_1$'s minimum budget.
Then, from $P_1$'s point of view, the best $r$ must come from the case when the two terms are equal,
so that
\[
     r^*\ =\ \frac{m_{i,j-1} - m_{i-1,j}}{1+m_{i,j-1}}.
\]
Consequently, we have
\begin{eqnarray*}
  m_{i,j}\ &=&\ r^* + m_{i-1,j}\ =\ (1-r^*)\,m_{i,j-1}\\
              &\ &\\
              &=&\ \frac{m_{i,j-1}\,(1+m_{i-1,j})}{1+m_{i,j-1}}.
\end{eqnarray*}
\end{proof}

\subsubsection*{Getting OBR from the Matrix.}
The lemmas above allow us to fill in the matrix $M$, using $O(1)$ time per entry, as follows:
\begin{enumerate}
  \item  Fill in the first row of $M$ by Lemma~\ref{lem:row-of-M}.
  \item  For $i = 2$ to $\myceil{T/2}$
  \begin{enumerate}
            \item Fill in $m_{i,i}$ by Lemma~\ref{lem:main-diagonal-of-M}, based on $m_{i-1,i}$ that has been computed.
            \item For $j = i+1$ to $\myceil{T/2}$, fill in $m_{i,j}$ by Lemma~\ref{lem:general-entry-of-M}, based on $m_{i-1,j}$ and $m_{i,j-1}$
                    that have been computed.
  \end{enumerate}
\end{enumerate}
Immediately, we have that for a $T$-turn game, the OBR countdown matrix $M$ can be obtained in $O(T^2)$ time.
Indeed, the following theorem gives a closed-form for each entry of $M$, so that we can obtain any entry on the fly in $O(1)$ time.
\begin{theorem}
For any $i \leq j$, the closed form of $m_{i,j}$ in the OBR countdown matrix $M$ is:
\[
  m_{i,j}\ =\ \frac{i\, (j-i+3)}{(j-i+1)(j+2)}.
\]
\end{theorem}
\begin{proof}
Let $F(i,j)$ denote the above closed form.
We show that $m_{i,j} = F(i,j)$ by induction on the row number $i$ (and with increasing order in the column number $j$).
First, the basis case is true, since $F(1,j) = (1 \times (j+2)) / (j \times (j+2))  = 1/j = m_{1,j}$.
Next, assume inductively that all entries in the $k$th row satisfy the closed form.
Then, when $i = k+1$ and $j = i = k+1$, we have:
\begin{eqnarray*}
    m_{k+1,k+1}\ &=&\  1 + m_{k,k+1}\ \,=\ 1 + F(k,k+1) \\
                        \ &=&\  1 + \frac{2k}{k+3}\ \ \ =\ \frac{3k+3}{k+3} \\
                        \ &=&\  F(k+1,k+1).
\end{eqnarray*}
As for $i = k+1$ and $j > i $, we have:
\begin{eqnarray*}
    &&m_{k+1,j}\\
    &=& \frac{m_{k+1,j-1}\,(1+m_{k,j})}{1+m_{k+1,j-1}}\ =\ \frac{F(k+1,j-1)\, (1+F(k,j))}{1+F(k+1,j-1)} \\
                      &\ & \\
                      &=&  \frac{(k+1)(j - k + 1)}{(j - k - 1)(j+1)} \times \\
                        &\ & \hspace{8ex} \left( 1 + \frac{k\, (j-k+3)}{(j-k+1)(j+2)} \right) \div \\
                        &\ & \hspace{16ex} \left( 1 + \frac{(k+1)(j-k+1)}{(j-k-1)(j+1)} \right) \\
                      &\ & \\
                      &=& \frac{ (k+1)[(j-k+1)(j+2) + k\,(j-k+3)]}{(j+2)[(j-k-1)(j+1) + (k+1)(j-k+1)]} \\
                      &\ & \\
                      &=& \frac{ (k+1)[(j-k+2)(j+2) - (j-2) + k\,(j-k+2) + k]}{(j+2)[(j-k)(j+1) - (j+1) + (k+1)(j-k)+ (k+1)]} \\
                      &\ & \\
                      &=& \frac{ (k+1)(j-k+2)[(j+2) + k -1]}{(j+2)(j-k)[(j+1) + (k+1) - 1]} \\
                      &\ & \\
                      &=& \frac{(k+1)(j-k+2)}{(j+2)(j-k)} \\
                      &\ &\\
                      &=&  F(k+1,j).
\end{eqnarray*}
Thus, all entries in the $(k+1)$th row satisfy the closed form, so that the inductive case is true.
The theorem thus follows.
\end{proof}

\newcommand{\obr}{\mathit{obr}}

\noindent
This gives the following corollary.
\begin{corollary}
Let $\obr(T)$ denote the optimal budget ratio for a $T$-turn game with set $\{ 0, 1\}$.  Then, we have:
\begin{enumerate}
\setlength{\itemsep}{3pt}
  \item $\obr(T) = m_{\myceil{T/2},\myceil{T/2}} = (3\,\myceil{T/2})/ (\myceil{T/2} + 2 ).$
  \item $\obr(T)$ is increasing.
  \item $\lim_{T \rightarrow \infty} \obr(T) = 3.$
\end{enumerate}
\end{corollary}

\subsection{Analysis for All-Pay Auctions}
In the discussion so far, we have implicitly assumed that the loser of a bid does not need to pay.  However, in practical settings such as our
presidential election example, if a player bids for $r$ units at some turn but loses the bid, it is often that some fraction $\alpha \in [0,1]$,
called \emph{all-pay ratio}, of the bid (i.e., $\alpha\,r$ units) has to be deducted from the budget for some overhead cost.

\medskip

\noindent
To compute the optimal budget ratio in this case, we can apply a similar analysis as before, where we define an OBR countdown matrix $N = [n_{i,j}]$,
such that $n_{i,j}$ denotes the OBR when the countdown values of $P_1$ and $P_2$ are $i$ and $j$, respectively.    Note that with such a charging scheme,
$P_2$ may use her omnipotent power to learn $P_1$'s bid $r$ beforehand, so that $P_2$ either pays $0$ when she does not want to win the bid, or pays $r+\epsilon$
for an arbitrary small $\epsilon > 0$ when she wants to win the bid.  Then, it is easy to argue that {\tt (i)} $n_{i,j} = +\infty$ for $i > j$, and (ii) the following
recurrences hold:

\begin{enumerate}
  \item $n_{0,j} = 0$ for all $j \geq 1$.
  \item $n_{i,i} = 1 + n_{i-1,i}$ for $i \geq 1$.
  \item $n_{i,j} = n_{i-1,j} + r^* =  n_{i,j-1}\,(1-r^*) + \alpha\,r^*$  for all $1 \leq i \leq j$, where $r^*$ denotes the optimal bid $P_1$
           should set at a turn when the countdown values of $P_1$ and $P_2$ are $i$ and $j$, respectively.
           The middle term corresponds to the case where $P_1$ wins the bid, so that $P_1$ still needs a budget of $n_{i-1,j}$ to win.
           The last term corresponds to the case where $P_2$ wins the bid, so that $P_1$ first needs to pay $\alpha\, r^*$ for the overhead cost,
           and needs a budget of $n_{i,j-1}\,(1-r^*)$ to win.
\end{enumerate}
Evaluating the third recurrence in the above, we get:
\[
    n_{i,j}\ =\ \frac{n_{i,j-1}n_{i-1,j} + n_{i,j-1} - \alpha\,n_{i-1,j}}{n_{i,j-1} + (1-\alpha)}.
\]
Based on the above recurrences, all entries $n_{i,j}$ with $i \leq j$ can be filled in a total of $O(T^2)$ time;  after that,
we can obtain the desired OBR as $n_{\myceil{T/2},\myceil{T/2}}$.  Moreover, we may derive closed form as before.
\begin{theorem}
For the standard all-pay auction, we have $\alpha = 1$, and the corresponding closed form, for $1\leq i \leq j$, is:
\[
   n_{i,j}\  =\ 1 + \frac{(i-1)(j-i+3)}{(j-i+1)(j+1)},
\]
indicating that the OBR in this case is increasing and has a limit of $4$.
\end{theorem}
\begin{proof}
\label{sec:proof-of-all-pay-closed-form}
We show that when $\alpha = 1$, the entry $n_{i,j}$ in the OBR countdown matrix in all-pay auction, for $1 \leq i \leq j$,
is equal to:
\[
    G(i,j) = 1 + \frac{(i-1)(j-i+3)}{(j-i+1)(j+1)},
\]
where $n_{i,j}$ satisfies the following recurrences:
\begin{enumerate}
  \item $n_{0,j} = 0$ for all $j \geq 1$, and $n_{i,j} = +\infty$ when $i > j$.
  \item $n_{i,i} = 1 + n_{i-1,i}$ for $i \geq 1$.
  \item $n_{i,j} = 1 + (n_{i-1,j}\,(n_{i,j-1} - 1)/n_{i,j-1})$ for all $1 \leq i < j$,
\end{enumerate}

\medskip
\noindent
To simplify the proof, we first observe the interesting fact that
\begin{enumerate}
  \item $G(i,j) = 1$ for $1 = i \leq j$.
  \item $G(i,j) = 1 + m_{i-1,j-1}$ for $2 \leq i \leq j$, where $m_{i,j}$ is the entry in the OBR countdown matrix $M$ for first-price auction.
\end{enumerate}
Furthermore, if we assume that $m_{0,j} = 0$, the recurrence about $m_{i,i}$ in Lemma~\ref{lem:main-diagonal-of-M} is generalized for $i \geq 1$,
and the recurrence about $m_{i,j}$ in Lemma~\ref{lem:general-entry-of-M} is generalized for $1 \leq i < j$, then, the above two facts can be unified as:
\[
    G(i,j) = 1 + m_{i-1,j-1} \mbox{\ \ \ \ for $1 \leq i \leq j$.}
\]
Thus, to show the correctness of the closed form, it is equivalent to proving that $n_{i,j} = 1 + m_{i-1,j-1}$ for all $1 \leq i \leq j$.

\medskip

\noindent
Our proof goes by induction on the row number $i$ (and with increasing order in the column number $j$).
First, the basis case corresponds to $i = 1$, in which we obtain $n_{1,1} = 1$
and $n_{1,j} = 1 + 0 = 1$, so that $n_{1,j} = 1 + m_{0,j-1}$ for all $j \geq i$.
Next, assume inductively that all entries in the $k$th row satisfy the closed form.

\medskip

\noindent
Then, when $i = k+1$ and $j = i = k+1$, we have:
\begin{eqnarray*}
    n_{k+1,k+1}\ &=&\  1 + n_{k,k+1}  \ \ \,=\ 1 + 1 + m_{k-1,k}  \\
                          \ &=&\  1 + m_{k,k},
\end{eqnarray*}
where the last equality follows from (the generalised) Lemma~\ref{lem:main-diagonal-of-M}.

\medskip

\noindent
As for $i = k+1$ and $j > i $, we have:
\begin{eqnarray*}
    n_{k+1,j} \ &=&\  1 + \frac{n_{k,j}\,(n_{k+1,j-1} - 1)}{n_{k+1,j-1}} \\
                      \ &\ &\ \\
                      \ &=&\  1 + \frac{(1+m_{k-1,j-1})\,m_{k,j-2}}{1+m_{k,j-2}} \\
                      \ &\ &\ \\
                      \ &=&\ 1 + m_{k,j-1},
\end{eqnarray*}
where the last equality follows from (the generalised) Lemma~\ref{lem:general-entry-of-M}.

\medskip

\noindent
Thus, all entries in the $(k+1)$th row satisfy the closed form, so that the inductive case is true.
This completes the proof of the induction, and consequently the closed form is correct.
\end{proof}

\subsection{Further Discussion} \label{sec:dis}
Suppose that the initial budget ratio $b_1/b_2$ is at least the OBR.
Then, $P_1$ can guarantee to win the game by bidding with the following strategy:
\begin{quotation}
\noindent
Initially, $P_1$ keeps track of $P_2$'s current budget $B$ as $b_2$, and sets the countdown values $i$ and $j$
of both players to be $\myceil{T/2}$.  At each turn, if the outcome is $1$,
$P_1$ uses countdown values $i$ and $j$ (of $P_1$ and $P_2$, respectively) to compute $r^* = m_{i,j} - m_{i-1,j}$, and bids with $r^* \times B$,
and aslo updates $i$, $j$, and $B$ depending who wins the bid;\footnote{%
In case $P_2$ wins the bid, but her bidding price is not disclosed, then $P_1$ simply assumes pessimistically that
$P_2$ wins the bid with $r^* + \epsilon$, and updates $B$ as $B - r^*$.}
else, if the outcome is $0$, $P_1$ updates $i$ and $j$ if needed.
\end{quotation}
The above strategy has the advantages that {\tt (i)} the optimal bid for each turn can be decided in $O(1)$ time, and {\tt (ii)} the updates for each turn take $O(1)$ time.
Morevoer, if $P_2$ does not have omnipotent power, then
$P_1$ may win the game with less budget, since $P_2$ may be overbidding (as she does not know $P_1$'s bid),
or cannot force a worst-case scenario (as she does not control the outcome).

\bigskip

\noindent
\subsubsection*{Bounded Falling Behind.} In most real-life scenarios, the initial budget ratio would be less than the OBR, so that there is no guaranteed winning.
However, $P_1$ may still want some guarantee of her performance, say, the difference $S_2 - S_1$ between the score of $P_2$ and the score of $P_1$
is at most $k$.   In such a \emph{handicapped} game, the OBR for a $T$-turn game becomes
$m_{\myceil{(T-k)/2},\myceil{(T+k)/2}}$, since if $P_1$ gets a score of $\myceil{(T-k)/2}$ beforehand, the score difference $S_2 - S_1$ is at most $k$,
while on the other hand, if $P_2$ gets a score of $\myceil{(T+k)/2}$ beforehand, $P_2$ can then use her omnipotent power to make the remaining outcomes all $0$,
so that $S_2 - S_1$ is at least $k+1$.

\newcommand{\comment}[1]{}

\section{Finding the Optimal Budget Ratio: Fixed Value}

In some contests such as sports, the winner of each component battle accumulates her partial score by a fixed value.
We model this scenario by fixing the value of each battle to 1.
Player $P_2$ becomes less powerful, as she can no longer control the common value for bidding of each turn.
In this section, we study the OBR for player $P_1$ under this model.

\subsection{Analysis for First-Price Auctions}
As before, we define an OBR countdown matrix $L = [\ell_{i,j}]$, such that
$\ell_{i,j}$ denotes the optimal budget ratio where countdown value of $P_1$ and $P_2$ are $i$ and $j$, respectively.
The following are some lemmas concerning the values of $\ell_{i,j}$ in different scenarios.

\begin{lemma} \label{lem:row-of-L}
For all $j$, $\ell_{1,j} = 1/j$.
\end{lemma}
\begin{proof}
Consider the time when the countdown values of $P_1$ and $P_2$ are $1$ and $j$, respectively.
Then, we have the following observations.

\begin{itemize}
\item
First, suppose that $P_1$ has a budget $b_1 =\$1/j$ and $P_2$ has a budget $b_2 = \$1$ .
Since the countdown value of $P_1$ is $1$, the game ends as soon as $P_1$ wins a turn.
Now, $P_1$ bids all her budget at each subsequent turn (which has outcome $1$).
To avoid $P_1$ from winning immediately, $P_2$ has to use her budget to win each turn, thus paying at least $1/j$ for each turn;
yet, $P_2$ can win at most $j-1$ turns, after which her budget drops below $1/j$, and $P_1$ would win the next turn.
Thus, $P_1$ prevails, which implies $\ell_{1,j} \leq 1/j$.

\item
Conversely, suppose that $P_1$ has a budget less than $\$1/j$.  Then, $P_2$ may use her omnipotent power to make all outcomes of
the next $j$ turns be $1$, where she bids $\$1/j$ for each turn.  Thus, $P_2$ will get $j$ extra values before $P_1$ gets $1$.
By Lemma~\ref{lem:countdown}, $P_2$ wins the game.  This implies $\ell_{1,j} \geq 1/j$.
\end{itemize}
Combining the above bounds on $\ell_{1,j}$ gives the desired bound $\ell_{1,j} = 1/j$.  The lemma follows.
\end{proof}

\begin{lemma} \label{lem:column-of-L}
For all $i$, $\ell_{i,1} = i$.
\end{lemma}
\begin{proof}
Since $P_2$ only needs to win one turn to win the whole game, to avoid this from happening,
$P_1$ must bid the same amount as $P_2$'s budget in each subsequent turn, for $i$ turns,
until she wins. This implies that $\ell_{i,1} = i$.
\end{proof}

\begin{lemma} \label{lem:general-entry-of-L}
When $2 \leq i , j$, the following recurrence holds:
\[
  \ell_{i,j}\ =\ \frac{\ell_{i,j-1}\,(1+\ell_{i-1,j})}{1+\ell_{i,j-1}}.
\]
\end{lemma}
\begin{proof}
The lemma follows from an analogous proof as that of Lemma~\ref{lem:general-entry-of-M}.
\end{proof}
\comment{
\begin{proof}
Consider the time when the countdown values of $P_1$ and $P_2$ are $i$ and $j$, respectively.
For $P_1$ to guarantee a winning, her minimum budget will be $\ell_{i,j}$, and such a budget
should allow her to post some optimal bid $r^*$ at this turn to her favour.  Precisely, we must have
\[
    r^*\ =\ \arg\min_r\, \max\, \{\, r + \ell_{i-1,j},  (1-r)\,\ell_{i, j-1}\, \},
\]
where the first term is the minimum budget that $P_1$ needs when $P_1$ gets the next value,
the second term is the minimum budget that $P_1$ needs when $P_2$ gets the next value;
the $\max$ operator corresponds to $P_2$'s choice to maximize $P_1$'s minimum budget.
Then, from $P_1$'s point of view, the best $r$ must come from the case when the two terms are equal,
so that
\[
     r^*\ =\ \frac{\ell_{i,j-1} - \ell_{i-1,j}}{1+\ell_{i,j-1}}.
\]
Consequently, we have
\[
  \ell_{i,j}\ =\ r^* + \ell_{i-1,j}\ =\ (1-r^*)\,\ell_{i,j-1}\ =\ \frac{\ell_{i,j-1}\,(1+\ell_{i-1,j})}{1+\ell_{i,j-1}}.
\]
\end{proof}
}

\subsubsection*{Getting OBR from the Matrix.}
The lemmas from the previous subsection allow us to fill in the matrix $L$, using $O(1)$ time per entry, as follows:
\begin{enumerate}
  \item  Fill in the first row of $L$ by Lemma~\ref{lem:row-of-L}.
  \item  Fill in the first column of $L$ by Lemma~\ref{lem:column-of-L}.
  \item  For $i = 2$ to $\lceil T/2 \rceil$
  \begin{enumerate}
            \item For $j = 2$ to $\lceil T/2 \rceil$, fill in $\ell_{i,j}$ by Lemma~\ref{lem:general-entry-of-L}, based on $\ell_{i-1,j}$ and $\ell_{i,j-1}$
                    that have been computed.
  \end{enumerate}
\end{enumerate}
Immediately, we have that for a $T$-turn game, the OBR countdown matrix $L$ can be obtained in $O(T^2)$ time.
Indeed, the following theorem gives a closed-form for each entry of $L$, so that we can obtain any entry on the fly in $O(1)$ time.
\begin{theorem}
For any $i, j$, the closed form of $\ell_{i,j}$ in the OBR countdown matrix $L$ is:
\[
  \ell_{i,j}\ =\ i/j,
\]
indicating that the OBR in this case is constantly 1.
\end{theorem}
\begin{proof}
Let $H(i,j)$ denote the above closed form.
We show that $\ell_{i,j} = H(i,j)$ by induction on the row number $i$ (and with increasing order in the column number $j$).
First, the basis case of first row is true, since $H(1,j) = 1/j = \ell_{1,j}$.
Also, the basis case of first column is also true, since $H(i,1) = i/1 = i = \ell_{i,1}$.
Next, assume inductively that entries $\ell_{i-1,j}$ and $\ell_{i,j-1}$ satisfy the closed form.
Then, considering entry $l_{i,j}$, we have:
\begin{eqnarray*}
    \ell_{i,j}\ &=&\  \frac{i/(j-1)\cdot(1+(i-1)/j)}{1+i/(j-1)}\ \\
                  &\ & \\
    					\ &=&\ \frac{i\cdot (i+j-1)/j}{i+j-1} \\
                  &\ & \\
                        \ &=&\  i/j \ = H(i,j).
\end{eqnarray*}
Thus, all entries satisfy the closed form, so that the inductive case is true.
The theorem thus follows.
\end{proof}

\subsection{Analysis for All-Pay Auctions}
The analysis for the case with all-pay auction, with $\alpha = 1$, is similar.
We define an OBR countdown matrix $Q = [q_{i,j}]$, such that
$q_{i,j}$ denotes the optimal budget ratio where countdown value of $P_1$ and $P_2$ are $i$ and $j$ in this scenario, respectively.
We have the following three lemmas.

\begin{lemma} \label{lem:row-of-L-allpay}
For all $j$, $q_{1,j} = 1$.
\end{lemma}
\begin{proof}
Consider the time when the countdown values of $P_1$ and $P_2$ are $1$ and $j$, respectively. Then, we have the following observations.

\begin{itemize}
\item
First, suppose that $P_1$ has the same budget as $P_2$.
Then, $P_1$ can ensure the winning by bidding all her budget at the first turn which implies $q_{1,j} \leq 1$.

\item
Conversely, suppose that $P_1$'s budget is less than $P_2$'s.  Then, each time $P_1$ bids $r$, $P_2$ can bid $r+\epsilon$, for an arbitrarily small $\epsilon > 0$, to win the turn.
In this way, $P_2$ will win all the subsequent turns, thus getting $j$ extra values before $P_1$ gets $1$.
By Lemma~\ref{lem:countdown}, $P_2$ wins the game.  This implies $q_{1,j} \geq 1$.
\end{itemize}
Combining the above bounds on $q_{1,j}$ gives the desired bound $q_{1,j} = 1$.  The lemma follows.
\end{proof}

\begin{lemma} \label{lem:column-of-L-allpay}
For all $i$, $q_{i,1} = i$.
\end{lemma}
\begin{proof}
Since $P_2$ only needs one more turn to win the game, to avoid this from happening,
$P_1$ must bid the same amount as $P_2$'s budget in each subsequent turn, for $i$ turns, to win the game.
This implies that $q_{i,1} = i$.
\end{proof}

\begin{lemma} \label{lem:general-entry-of-L-allpay}
When $2 \leq i , j$, the following recurrence holds:
\[
  q_{i,j}\ =\ 1\ -\ \frac{q_{i-1,j}}{q_{i,j-1}}\ +\ q_{i-1,j}.
\]
\end{lemma}
\begin{proof}
Consider the time when the countdown values of $P_1$ and $P_2$ are $i$ and $j$, respectively.
For $P_1$ to guarantee a winning, her minimum budget will be $q_{i,j}$, and such a budget
should allow her to post some optimal bid $r^*$ at this turn to her favour.  Precisely, we must have
\[
    r^*\ =\ \arg\min_r\, \max\, \{\, r + q_{i-1,j},\  r +(1-r)\,q_{i, j-1}\, \},
\]
where the first term is the minimum budget that $P_1$ needs when $P_1$ gets the next value,
the second term is the minimum budget that $P_1$ needs when $P_2$ gets the next value;
the $\max$ operator corresponds to $P_2$'s choice to maximize $P_1$'s minimum budget.
Then, from $P_1$'s point of view, the best $r$ must come from the case when the two terms are equal,
so that
\[
     r^*\ =\ 1 - \frac{q_{i-1,j}}{q_{i,j-1}}.
\]
Consequently, we have
\begin{eqnarray*}
  q_{i,j}&=& r^*\ +\ q_{i-1,j}\ =\ 1\ -\ \frac{q_{i-1,j}}{q_{i,j-1}}\ +\ q_{i-1,j}.
\end{eqnarray*}
\end{proof}

The above lemmas allow us to fill in the matrix $Q$ in the same way as we fill in $L$ as follows:
\begin{enumerate}
  \item  Fill in the first row of $Q$ by Lemma~\ref{lem:row-of-L-allpay}.
  \item  Fill in the first column of $Q$ by Lemma~\ref{lem:column-of-L-allpay}.
  \item  For $i = 2$ to $\lceil T/2 \rceil$
  \begin{enumerate}
            \item For $j = 2$ to $\lceil T/2 \rceil$, fill in $q_{i,j}$ by Lemma~\ref{lem:general-entry-of-L-allpay}, based on $q_{i-1,j}$ and $q_{i,j-1}$
                    that have been computed.
  \end{enumerate}
\end{enumerate}
Immediately, we have that for a $T$-turn game, the OBR countdown matrix $Q$ can be obtained in $O(T^2)$ time.
Indeed, the following theorem gives a closed-form for each entry of $Q$, so that we can obtain any entry on the fly in $O(1)$ time.
\begin{theorem}
For any $i, j$, the closed form of $q_{i,j}$ in the OBR countdown matrix $Q$ is:
\[
  q_{i,j}\ =\ \frac{i+j-1}{j},
\]
indicating that the OBR in this case is increasing and has a limit of 2.
\end{theorem}
\begin{proof}
Let $J(i,j)$ denote the above closed form.
We show that $q_{i,j} = J(i,j)$ by induction.
First, the basis case of first row is true, since $J(1,j) = (1+j-1)/j = 1 = q_{1,j}$.
Also, the basis case of first column is also true, since $J(i,1) = (i+1-1)/1 = i = q_{i,1}$.
Next, assume inductively that entries $q_{i-1,j}$ and $q_{i,j-1}$ satisfy the closed form.
Then, considering entry $q_{i,j}$, we have:
\begin{eqnarray*}
    q_{i,j}\ &=&\  1\ -\ \left(\frac{i+j-2}{j}\right) \div \left(\frac{i+j-2}{j-1}\right)\ + \\
               &&\\
               && \hspace{8ex} \left(\frac{i+j-2}{j}\right)\ \\
               &&\\
    					\ &=&\ 1\ -\ \frac{j-1}{j}\ +\ \frac{i+j-2}{j} \ \\
			&&\\		
                        \ &=&\  \frac{i+j-1}{j}\ =\ J(i,j).
\end{eqnarray*}
Thus, all entries satisfy the closed form, so that the inductive case is true.
The theorem follows.
\end{proof}

\begin{remark}
A discussion similar to Section~\ref{sec:dis} for the corresponding bidding strategies and bounded falling behind in the case of a fixed value can be analogously conducted here.
\end{remark} 
\section{Conclusions and Future Work}
We model and study budget-constrained multi-battle contests as extensive form zero-sum games.
When each turn is a first-price or all-pay auction with fixed value 1 or value set $\{0,1\}$,
we focus on analyzing the 2-player optimal budget ratio that guarantees player~$P_1$'s winning (or losing a bounded number of valued turns)
against an omnipotent player~$P_2$.
We give efficient dynamic programs to find the optimal budget ratio and the corresponding series of bidding strategies.

In our current results, we mainly ensure player~$P_1$'s winning when her budget is not less than the optimal budget ratio;
when allowing player~$P_1$ to lose at most a bounded number of valued battles,
we in some sense relax her budget requirement to be less since the optimal ratio is smaller now.
Yet, this is not the only way to look at cases when the budget ratio is less than the optimal one.
Another direction is to ask about bidding strategies that lead to more situations (i.e., a series of values chosen from the value set)
ending up in winning since now a player with a moderate budget may not always win.
Inspired by the budget ratio analysis here,
there might exist a dynamic program to find such bidding strategies that maximize the winning situations over all possible situations.
Other immediate things to consider include lessening player~$P_2$'s power by making the number of valued turns known to player~$P_1$.
Generally, considering multiple players in budget-constrained multi-battle contests is challenging,
and may take other approaches to analyze. 

\bibliographystyle{abbrv}

\end{document}